\documentclass[conference]{IEEEtran}
\usepackage{amssymb}
\usepackage{amsfonts}
\usepackage{graphicx}
\usepackage{array}
\usepackage{cite}
\usepackage{psfrag}
\usepackage{stfloats}
\usepackage{amsmath}
\usepackage{url}
\usepackage{times}
\usepackage{subfigure}

\begin{document}

\title{Collision Codes: Decoding Superimposed BPSK Modulated Wireless Transmissions}

\author{Chuan Heng Foh, Jianfei Cai, and Jalaluddin Qureshi\\
School of Computer Engineering\\
Nanyang technological University, Singapore\\
\{aschfoh, asjfcai, jala0001\}@ntu.edu.sg}%

\maketitle

\begin{abstract}
The introduction of physical layer network coding gives rise to the
concept of turning a collision of transmissions on a wireless
channel useful. In the idea of physical layer network coding, two
synchronized simultaneous packet transmissions are carefully encoded
such that the superimposed transmission can be decoded to produce a
packet which is identical to the bitwise binary sum of the two
transmitted packets. This paper explores the decoding of
superimposed transmission resulted by multiple synchronized
simultaneous transmissions. We devise a coding scheme that achieves
the identification of individual transmission from the synchronized
superimposed transmission. A mathematical proof for the existence of
such a coding scheme is given.
\end{abstract}

\section{Introduction}
The concept of physical-layer network coding (PNC) first introduced
in 2006 by Zhang \emph{et al.} reveals the idea of decoding a
transmission collision on a wireless channel \cite{PNC}. This
concept directly challenges the \emph{traditional rule} that a
collided transmission on a wireless channel is undecodable. In this
pioneering work, it has been demonstrated that a collision of two
simultaneous wireless transmissions can be turned into a useful
transmission. In brief, two simultaneous wireless transmissions that
are added together at the electromagnetic wave level can be decoded
and mapped to produce an outcome such that the relationship between
the transmitted and the decoded binary information follows the
exclusive-or (XOR) principle.

The concept of PNC is found to have potential in enhancing
performance of the current wireless networks. In the original paper
describing PNC in \cite{PNC}, the concept is demonstrated to be
effective for further performance enhancement of network coding
operation in wireless environment~\cite{Ahlswede00,Katti06}.
In~\cite{pnc_capacity} the suitability of PNC is shown to improve
the throughput capacity of a random wireless network by a fixed
factor. Modulation and mapping schemes are proposed to allow for the
decoding of a collision from two simultaneous transmissions.
However, the scheme is only suitable for decoding a superimposed
transmission consisting only two simultaneous transmissions within,
and this greatly limits its applications.

Recently, Durvy \emph{et al.}~\cite{Durvy07} applies the idea of
decoding superimposed signals in PNC to improve the reliability of
wireless broadcasting. While traditional approaches often attempt to
avoid collisions~\cite{Pagani97,Tang01}, in~\cite{Durvy07}, the
authors propose modified ARQ operation to invite collision of
multiple acknowledgement (ACK) transmission and then design a scheme
to decode the collided ACK transmissions. Their idea is that upon
receiving a broadcast transmission, each receiver detecting the
transmission replies with an ACK transmission. These simultaneous
ACK transmissions will cause a collision. Using the concept of PNC,
decoding of the superimposed ACK transmissions is performed to
identify the ACK transmitters. Their method assumes synchronous
simultaneous ACK transmissions and the capability of precise
detection of signal energy. Since simultaneous ACK transmissions
appear after the completion of a broadcast transmission which is a
common event, simultaneous ACK transmissions may be considered
synchronized to a certain extent. However, the requirement of
precise detection of signal energy for decoding introduces
difficulty in the practical design.

We here define \emph{collision decoding} to be the capability of the
identification of individual transmission from the superimposed
transmissions. In other words, it is the capability for the
collision decoder to tell which set of stations has transmitted and
which has no transmitted in a superimposed transmission. A
particular application of this capability has been shown
in~\cite{Durvy07} to be beneficial in protocol design and
enhancement. However in~\cite{Durvy07}, the need for precise
detection of signal energy for decoding has reduced the
practicability of the proposal. In this paper, we devise a coding
scheme for collision decoding, that particularly addresses the
challenge of signal energy detection. Our scheme overcomes the
shortcoming of the method given in~\cite{Durvy07}. We also
demonstrate with a mathematical proof the existence of such a coding
scheme.

The rest of the paper is organized as follows. Section~\ref{sec:background}
revisits the PNC technique and describes our proposed scheme. In
Section~\ref{sec:proof}, we give the mathematical proof that our proposed
coding scheme allows unique identification of different receiver combinations.
Finally, important conclusion is drawn in Section~\ref{sec:conclusion}.

\section{Decoding a Collision}\label{sec:background}

\subsection{Physical-Layer Network Coding}

The concept of PNC was first introduced by Zhang \emph{et al.}
in~\cite{PNC}. The authors proposed a frame-based Decode-and-Forward
(and amplify if necessary) strategy in packet forwarding. In their
scenario, two neighboring nodes transmit simultaneously to a common
receiver. Assuming perfect transmission synchronization in physical
layer, based on the additive nature of simultaneously arriving
electromagnetic (EM) waves, the receiver detects the added signal of
the two transmitted modulated signals. Using a suitable mapping
scheme, they show that for certain modulation schemes, there exists
a mapping scheme such that the relationship between the two
transmitted binary bits and the decoded binary bit follows the XOR
principle.

We here revisit the PNC operation~\cite{PNC} and its mapping scheme
to achieve the XOR principle. Consider two senders, $N_1$ and $N_3$,
and a common receiver $N_2$. Let $s_1$ and $s_3$ be the binary bit
transmitted by $N_1$ and $N_3$ at a particular time respectively,
and $s_2$ be the decoded binary bit. Based on BPSK modulation, we
have
\begin{equation}
\begin{array}{ll}
r_2(t) = a_1 \cos(\omega t) + a_3 \cos(\omega t) = (a_1+a_3) \cos(\omega t) \\
\end{array}
\end{equation}
where $r_2(t)$ is the received signal, $a_1$ and $a_3$ are the
transmitted amplitudes, and $\omega$ is the carrier frequency. For
BPSK, we have $a_j=2s_j-1$~\cite{PNC}. At $N_2$, a scheme (see
Table~\ref{table:PNC}) that maps a strong energy signal to binary 0
(i.e. $|a_1+a_3|=2$) and a weak energy signal (i.e. $a_1+a_3=0$) to
binary 1 can be applied which gives
\begin{equation}
s_2 = s_1 \oplus s_3.
\end{equation}

\begin{table}
\centering \caption{\label{table:PNC} The PNC mapping for two transmitting
nodes using BPSK.}
\begin{tabular}{|c|c|c|c|c|c|}
\hline \multicolumn{4}{|c|}{Modulation mapping} &
\multicolumn{2}{c|}{Demodulation mapping} \\
\multicolumn{4}{|c|}{at $N_1$ and $N_3$} &
\multicolumn{2}{c|}{at $N_2$} \\
%\multicolumn{4}{|c|}{Modulation mapping at $N_1$ and $N_3$} &
%\multicolumn{2}{c|}{Demodulation mapping at $N_2$} \\
\hline \multicolumn{2}{|c}{Input} & \multicolumn{2}{|c}{Output} &
\multicolumn{1}{|c|}{Input} & \multicolumn{1}{c|}{Output} \\
\hline $s_1$ & $s_3$ & $a_1$ & $a_3$ & $a_1$ + $a_3$ & $s_2$ \\
\hline 1 & 1 & 1 & 1 & 2 & 0 \\
\hline 0 & 1 & -1 & 1 & 0 & 1 \\
\hline 1 & 0 & 1 & -1 & 0 & 1 \\
\hline 0 & 0 & -1 & -1 & -2 & 0 \\
\hline
\end{tabular}
\end{table}

\subsection{Collision of Multiple BPSK Modulated Transmissions}

Considering a collision consisting of an arbitrary number of
simultaneous transmissions, according to Table~\ref{table:PNC},
retaining the XOR principle for the decoding requires precise
detection of signal energy for the PNC mapping, which reduces its
robustness. In~\cite{Durvy07}, the authors proposed a scheme to deal
with decoding of multiple simultaneous transmissions. The scheme
proposes bit sequences of $N+1$ bits that is suitable for decoding a
collision of up to $N$ simultaneous transmissions. Likewise, the
main limitation of the proposed scheme is that it also requires
precision of power level differentiation in the decoding procedure.
Precisely, comparisons of analog received signals are needed for the
operation, and the authors propose use of a delay line to store
analog signal for the comparison purpose. If we relax the needs for
power level differentiation, based on the additive nature of EM
waves, considering the BPSK modulation scheme, it is not difficult
to see that the BPSK demodulation process follows the majority
principle.

To illustrate this, let us consider a simple case of three
transmitters, $N_i \; (i \in\{1,2,3\})$, and one common receiver,
$N_0$. Let $a_i$ be the amplitude of the BPSK signals corresponding
to the transmitted binary information, and $a_0$ be the detected
amplitude of the BPSK signals. Based on the additive nature of EM
waves, we have $a_0 = a_1+a_2+a_3$. Considering a common design of a
BPSK demodulator with a matched filter and a detection device, the
demodulator produces 1 if $a_0 > 0$ and 0 otherwise.
Table~\ref{table:BPSK} exhausts all possible inputs and shows the
relationship between the input and the output binary information. As
can be seen, the relationship follows the majority principle.

\begin{table}
\begin{center}
\caption{\label{table:BPSK} An example of BPSK Demodulation for three
transmitters.}
\begin{tabular}{|c|c|c|c|c|c|c|c|c|c|}
\hline \ & $N_1$ & 0 & 0 & 0 & 0 & 1 & 1 & 1 & 1 \\
\cline{2-10}
\ Input & $N_2$ & 0 & 0 & 1 & 1 & 0 & 0 & 1 & 1 \\
\cline{2-10}
\ & $N_3$ & 0 & 1 & 0 & 1 & 0 & 1 & 0 & 1 \\
\hline
\ Output & $N_0$ & 0 & 0 & 0 & 1 & 0 & 1 & 1 & 1 \\
\hline
\end{tabular}
\end{center}
\end{table}

A closer examination of BPSK shows that any pair of inputs that
holds different binary information is offset when added at the EM
level. As a result, the remaining input decides the binary outcome.
In the case of a tie, since the detected energy fails to reach a
threshold after the matched filter, a consistent conclusion will be
made. Without loss of generality, we assume it to be binary 0.

\subsection{Our Proposed Scheme for Collision Decoding}

Using the majority principle, we design a coding scheme that enables the common
receiver of a collided transmission to tell the presence of individual
transmission involved in the collision. We first illustrate the application of
collision decoding using multicast application, followed by the proposed coding
scheme.

Consider a wireless network consisting of a basestation and a collection of
stations within wireless coverage of the basestation. For the multicast
transmission, the basestation broadcasts the transmission on the wireless
channel and only the intended receivers will process the transmission. In our
design, the basestation also embeds information in its data transmission to
tell each individual intended receiver of its unique identifier.

It is possible that not all intended receivers detect the transmission due to,
for example noise. If an intended receiver detects the broadcast transmission,
it first extracts its identifier embedded by the basestation, say $i$, and
immediately replies an ACK with a predefined unique bitstream, $\mathbf{s}_i$,
as part of ACK.

With the immediate replies of ACK transmissions from multiple
receivers, a collision of ACK transmissions occurs. The basestation
decodes the superimposed transmission using a BPSK demodulator and
as discussed in the previous subsection, the decoded bitstream, say
$\mathbf{v}$, will follow majority principle. We shall show that
there exists a coding scheme such that the basestation produces a
unique bitsream for a particular combination of the receivers'
bitstreams. This unique bitsream enables the basestation to identify
whether a particular intended receiver has replied ACK indicating
the success delivery of the multicast packet to that receiver.

Consider the number of receiver, $N$, to be an odd number, and a bitstream has
$V$ bits. Let $R=\frac{N+1}{2}$. We first construct a binary matrix of size $N
\times V$ such that each column contains exactly $R$ of binary 1 and $R-1$ of
binary 0, with a unique permutation. Exhausting all permutations given $R$
number of binary 1 and $R-1$ number of binary 0 produces ${N \choose R}$ unique
patterns, and hence $V = {N \choose R}$. With this construction, the binary
matrix holds $N$ number of unique $V$-bit bitstreams, each of which will be
assigned to a receiver. While our scheme uses bit sequences of $N \choose R$
bits, this method remains adequate to support up to 15 stations as the required
number of bits is below 8000 bits (or 1000 bytes), which is suitable for
existing protocols to carry.

We give an example of $N=3$ in the following. According to our scheme, the
binary matrix for $N=3$, $\mathbf{M}_3$, can be constructed as
\[
\mathbf{M}_3 = \left[\begin{array}{ccc} 1 & 1 & 0 \\ 1 & 0 & 1 \\ 0 & 1 & 1
\end{array}\right].
\]

Each row in the binary matrix represents the unique bitstream for each receiver
to transmit in its ACK. In Table~\ref{table:3nodes}, we show the assignment of
bitstreams to the three receivers, $N_1, N_2, N_3$. Moreover, we provide the
decoded bitstreams at the basestation given all possible combinations of ACK
replies. In this example, we see that a particular combination of receivers can
be uniquely identified by the decoded bitstream.

\begin{table}
\begin{center}
\caption{\label{table:3nodes} The decoded bitstreams (based on the majority
principle) of different receiver combinations for $N=3$.}
\begin{tabular}{|c|c|}
\hline

Receiver Combination    &{Decoded Bitstream} \\ \hline

$N_1$   &110 \\ \hline

$N_2$   &101 \\ \hline

$N_3$   &011 \\ \hline

($N_1$,$N_2$)   &100 \\ \hline

($N_1$,$N_3$)   &010 \\ \hline

($N_2$,$N_3$)   &001 \\ \hline

($N_1$,$N_2$,$N_3$)   &111 \\ \hline
%None    &000 \\ \hline
\end{tabular}
\end{center}
\end{table}

\section{The Proof} \label{sec:proof}

\newtheorem{theorem}{Theorem}
\newtheorem{lemma}{Lemma}
\newtheorem{claim}{Claim}
\newtheorem{property}{Property}
\newtheorem{proposition}[theorem]{Proposition}
\newtheorem{corollary}[theorem]{Corollary}

\newenvironment{definition}[1][Definition]{\begin{trivlist}
\item[\hskip \labelsep {\bfseries #1}]}{\end{trivlist}}
\newenvironment{example}[1][Example]{\begin{trivlist}
\item[\hskip \labelsep {\bfseries #1}]}{\end{trivlist}}
\newenvironment{remark}[1][Remark]{\begin{trivlist}
\item[\hskip \labelsep {\bfseries #1}]}{\end{trivlist}}

\newcommand{\qed}{\nobreak \ifvmode \relax \else
      \ifdim\lastskip<1.5em \hskip-\lastskip
      \hskip1.5em plus0em minus0.5em \fi \nobreak
      \vrule height0.75em width0.5em depth0.25em\fi}

In this section, we show that our proposed coding scheme guarantees the
uniqueness of the decoded bitstream for a particular receiver combination for
any odd value of $N$. In case that the number of receivers $N$ is an even
number, we can use $N+1$ codes for the same purpose, and the additional
bitstream in the codes will not be used.

Consider $N$ stations, based on our design, we can construct a coefficient
matrix $\mathbf{M_N}$, where each of its elements is either $+1$ or
$-1$\footnote{Here we use $+1$ and $-1$ to replace $1$ and $0$ for better
illustration.}. Let $G=\{1,...,N\}$ be the set containing all stations and
$R=\frac{N+1}{2}$. According to our design, each of the column in
$\mathbf{M_N}$ contains exactly $R$ number of $+1$ and $R-1$ number of $-1$. In
other words, for a given column $c$, $\sum_{r \in G} \mathbf{M_N}(r,c)=1$.
Based on this design, exhaustive permutation of $+1$ and $-1$ for a column
construction gives ${N \choose R}$ unique patterns. The coefficient matrix
$\mathbf{M_N}$ that holds non-repetitive patterns of columns thus has a size of
$N \times {N \choose R}$.

We define $F(G_0,c)=\sum_{r \in G_0} \mathbf{M_N}(r,c)$ where $G_0 \subseteq
G$. In other words, the function $F(G_0,c)$ gives the sum of the values
corresponding to a given column $c$ and a particular collection of stations
$G_0$. Let $\Omega_+(G_0,c)$ (resp. $\Omega_-(G_0,c)$) denote the function that
counts the number of $+1$ (resp. $-1$) corresponding to the the column $c$ and
the collection of stations given in $G_0$. By definition, we have
\begin{equation}\label{eq:omega}
F(G_0,c) = \Omega_+(G_0,c)-\Omega_-(G_0,c).
\end{equation}

We further use the notation $|G_0|$ to denote the cardinality of the set $G_0$,
$G_0'=G \setminus G_0$ to denote the complementary set of $G_0$, and
$\overrightarrow{G_0}$ to denote a vector whose elements are given by
\begin{equation} \label{eq:gvec}
\overrightarrow{G_0}(c) = \left\{
\begin{array}{ll}
1, & F(G_0,c) \geq 1 \\
0, & F(G_0,c) < 1 \\
\end{array}
\right.
\end{equation}
where $c=1,2,...,{N \choose R}$ and $R=\frac{N+1}{2}$.

Based on the above definition, we first have the following properties about
$F(\cdot)$.

\begin{claim} \label{claim:op-plus}
Let $G_1, G_2 \subseteq G$ and $G_1 \cap G_2 = \emptyset$. If $H = G_1 \cup
G_2$, then for all $c$, $F(H,c) = F(G_1,c) + F(G_2,c)$.
\end{claim}

\begin{claim} \label{claim:op-minus}
Let $G_1 \subset G_2 \subseteq G$. If $H = G_2 \setminus G_1$, then for all
$c$, $F(H,c) = F(G_2,c) - F(G_1,c)$.
\end{claim}

The above claims can be easily established by set operations. By definition,
$F(H,c) = \sum_{r \in H} \mathbf{M_N}(r,c)$. Given that $H = G_1 \cup G_2$, we
obtain
\[
\begin{array}{ll}
F(H,c) & = \sum\limits_{r \in (G_1 \cup G_2)} \mathbf{M_N}(r,c) \\
%       & = \sum\limits_{r \in G_1} \mathbf{M_N}(r,c) + \sum\limits_{r \in G_2} \mathbf{M_N}(r,c) \\
%       & \;\; - \sum\limits_{r \in (G_1 \cap G_2)} \mathbf{M_N}(r,c) \\
       & = \sum\limits_{r \in G_1} \mathbf{M_N}(r,c) + \sum\limits_{r \in G_2} \mathbf{M_N}(r,c)
\end{array}
\]
since $G_1 \cap G_2 = \emptyset$. With the above result, we have established
claim \ref{claim:op-plus}. Claim~\ref{claim:op-minus} can be established with
the same approach.

Since in our design, each column in $\mathbf{M_N}$ contains exactly
$\frac{N+1}{2}$ number of $+1$ and $\frac{N-1}{2}$ number of $-1$, then by
definition
\begin{equation}\label{eq:F}
%%\begin{array}{l}
\Omega_+(G,c) = \frac{N+1}{2} \\
\Omega_-(G,c) = \frac{N-1}{2} \\
F(G,c)=1.
%%\end{array}
\end{equation}

The above gives the following lemmas.

\begin{lemma} \label{lemma:2a}
Let $G_1 \subset G$ where $|G_1|=g$ and $g$ is an odd integer $\leq N-1$. There
exists a column $c$ in $\mathbf{M_N}$ such that $F(G_1,c)=1$.
\end{lemma}
\begin{proof}
By definition, $\mathbf{M_N}$ holds exhaustive patterns of columns with exactly
$R$ number of $+1$ and $R-1$ number of $-1$. In other words, any column with
exactly $R$ number of $+1$ and $R-1$ number of $-1$ is a column of
$\mathbf{M_N}$.

Let $G_1 \subset G$ where $|G_1|=g$ and $g$ is an odd integer $\leq N-1$. We
create a column $c$ such that $\Omega_+(G_1,c) = \frac{g+1}{2}$,
$\Omega_-(G_1,c) = \frac{g-1}{2}$ and
$\Omega_+(G_1',c)=\Omega_-(G_1',c)=\frac{N-g}{2}$. Since $G=G_1 \cup G_1'$ and
$G_1 \cap G_1' = \emptyset$, we yield
\[
\begin{array}{l}
\Omega_+(G,c) = \Omega_+(G_1,c) + \Omega_+(G_1',c) = \frac{N+1}{2}=R
\\
\Omega_-(G,c) = \Omega_-(G_1,c) + \Omega_-(G_1',c) =
\frac{N-1}{2}=R-1 \\
\end{array}
\]
which shows that $c$ is a column of $\mathbf{M_N}$. Since $\Omega_+(G_1,c) =
\Omega_-(G_1,c)+1$, by (\ref{eq:omega}) we thus have $F(G_1,c)=1$.

\vspace{-4.5mm}

\end{proof}

\vspace{4.5mm}

\begin{lemma} \label{lemma:2b}
Let $G_1 \subset G$ where $|G_1|=g$ and $g$ is a nonzero even integer $\leq
N-1$. There exists a column $c$ in $M_N$ such that $F(G_1,c)=0$.
\end{lemma}
\begin{proof}
Let $G_1 \subset G$ where $|G_1|=g$ and $g$ is a nonzero even integer. We
create a column $c$ such that $\Omega_+(G_1,c) = \Omega_-(G_1,c) =
\frac{g}{2}$, $\Omega_+(G_1',c)=\frac{N-g+1}{2}$, and
$\Omega_-(G_1',c)=\frac{N-g-1}{2}$. Since $G=G_1 \cup G_1'$ and $G_1 \cap G_1'
= \emptyset$, we yield
\[
\begin{array}{l}
\Omega_+(G,c) = \Omega_+(G_1,c) + \Omega_+(G_1',c) = \frac{N+1}{2}=R
\\
\Omega_-(G,c) = \Omega_-(G_1,c) + \Omega_-(G_1',c) =
\frac{N-1}{2}=R-1 \\
\end{array}
\]
which shows that $c$ is a column of $\mathbf{M_N}$. Since $\Omega_+(G_1,c) =
\Omega_-(G_1,c)$, by (\ref{eq:omega}) we thus have $F(G_1,c)=0$.

\vspace{-4.5mm}

\end{proof}

\vspace{4.5mm}

\begin{proposition} \label{prop:code}
Consider a certain system with $N$ stations where $N \geq 1$. For any $G_1,G_2
\subset G$ and $G_1,G_2 \neq \emptyset$, $\overrightarrow{G_1} =
\overrightarrow{G_2}$ iif $G_1 = G_2$.
\end{proposition}
\begin{proof}
By definition given in (\ref{eq:gvec}), we can easily see that if $G_1 = G_2$,
then $\overrightarrow{G_1} = \overrightarrow{G_2}$. In the following, we shall
prove that if $G_1 \neq G_2$, then $\overrightarrow{G_1}
\neq\overrightarrow{G_2}$.

We will examine two cases: (i) $|G_1|$ is an odd integer; (ii) $|G_1|$ is an
even integer. Without loss of generality, we assume that $|G_1| \geq |G_2|$.

{Case (i)}: $|G_1|$ is an odd integer. Let $g=|G_1|$. Based on
lemma~\ref{lemma:2a}, there exists a column $c$ in $\mathbf{M_N}$ such that
\begin{equation} \label{eq:case1-g1}
F(G_1,c)=1,
\end{equation}
and hence according to (\ref{eq:gvec}), we also have
\[
\overrightarrow{G_1}(c)=1.
\]

This column will have the following property
\begin{equation}\label{eq:case1-omega}
\begin{array}{l}
\Omega_+(G_1,c) = \frac{g+1}{2} \\
\Omega_-(G_1,c) = \Omega_+(G_1,c)-1 = \frac{g-1}{2} \\
\Omega_+(G_1',c) = \Omega_-(G_1',c) = \frac{N-g}{2}
\end{array}
\end{equation}
with any permutation.

If $G_1 \neq G_2$, there exists $K_1 = G_1 \setminus G_2 \neq \emptyset$ and
$K_2 = G_2 \setminus G_1$ where $|K_1|=k_1\geq 1$ and $|K_2|=k_2 \geq 0$. It is
clear that $K_1 \subseteq G_1$, $K_2 \subseteq G_1'$ imply
\[
1 \leq k_1 \leq g, k_2 \leq N-g.
\]
where $g$ is an odd integer and $N-g$ is an even integer.

With the above condition, there exists a permutation within $c$ such that
\[
\begin{array}{l}
\Omega_+(K_1,c) = \lfloor\frac{k_1}{2}\rfloor + 1 \leq \frac{g+1}{2} \\
\Omega_-(K_1,c) = \lfloor\frac{k_1-1}{2}\rfloor \leq \frac{g-1}{2} \\
\Omega_+(K_2,c) = \lfloor\frac{k_2}{2}\rfloor \leq \frac{N-g}{2} \\
\Omega_-(K_2,c) = \lfloor\frac{k_2+1}{2}\rfloor \leq \frac{N-g}{2} \\
\end{array}
\]
that confirms $K_1 \subseteq G_1$, $K_2 \subseteq G_1'$ given
(\ref{eq:case1-omega}). The immediate result gives
\begin{equation} \label{eq:case1-k1}
F(K_1,c) = \left\{
\begin{array}{ll}
1, & k_1=1,3,... \\
2, & k_1=2,4,... \\
\end{array}\right.
\end{equation}
and
\begin{equation} \label{eq:case1-k2}
F(K_2,c) = \left\{
\begin{array}{ll}
-1, & k_2=1,3,... \\
 0, & k_2=0,2,4,... \\
\end{array}\right.
\end{equation}
which yield
\begin{equation} \label{eq:case1-k}
F(K_1,c) - F(K_2,c) \geq 1.
\end{equation}

Since $G_2 = (G_1 \setminus K_1) \cup K_2$, with (\ref{eq:case1-g1}) and
(\ref{eq:case1-k}), we obtain
\[
F(G_2,c) = F(G_1,c) - F(K_1,c) + F(K_2,c) \leq 0
\]
and hence $\overrightarrow{G_2}(c)=0$. Since $\overrightarrow{G_1}(c)=1$, there
exists a column in $\mathbf{M_N}$ such that $\overrightarrow{G_2}(c) \neq
\overrightarrow{G_2}(c)$, which is sufficient to show that
$\overrightarrow{G_1}\neq\overrightarrow{G_2}$.

{Case (ii)}: $|G_1|$ is an even integer. Let $g=|G_1|$. Based on
lemma~\ref{lemma:2b}, there exists a column $c$ in $\mathbf{M_N}$ such that
\begin{equation} \label{eq:case2-g1}
F(G_1,c)=0
\end{equation}
and hence according to (\ref{eq:gvec}), we also have
\[
\overrightarrow{G_1}(c)=0.
\]

This column will have the following property
\begin{equation}\label{eq:case2-omega}
\begin{array}{l}
\Omega_+(G_1,c) = \Omega_-(G_1,c) = \frac{g}{2} \\
\Omega_+(G_1',c) = \frac{N-g+1}{2} \\
\Omega_-(G_1',c) = \frac{N-g-1}{2}
\end{array}
\end{equation}
with any permutation.

If $G_1 \neq G_2$, there exist $K_1 = G_1 \setminus G_2 \neq \emptyset$ and
$K_2 = G_2 \setminus G_1$ where $|K_1|=k_1\geq 1$ and $|K_2|=k_2 \geq 0$. It is
clear that $K_1 \subseteq G_1$, $K_2 \subseteq G_1'$ imply
\[
1 \leq k_1 \leq g, k_2 \leq N-g.
\]
where $g$ is an nonzero even integer and $N-g$ is an odd integer.

We first exclude the condition where $k_2=0$. With the above condition, there
exists a permutation within $c$ such that
\[
\begin{array}{l}
\Omega_+(K_1,c) = \lfloor\frac{k_1}{2}\rfloor \leq \frac{g}{2} \\
\Omega_-(K_1,c) = \lfloor\frac{k_1+1}{2}\rfloor \leq \frac{g}{2} \\
\Omega_+(K_2,c) = \lfloor\frac{k_2}{2}\rfloor + 1 \leq \frac{N-g+1}{2} \\
\Omega_-(K_2,c) = \lfloor\frac{k_2-1}{2}\rfloor \leq \frac{N-g-1}{2} \\
\end{array}
\]
that confirms $K_1 \subseteq G_1$, $K_2 \subseteq G_1'$ given
(\ref{eq:case2-omega}). The immediate result gives
\begin{equation}
F(K_1,c) = \left\{
\begin{array}{ll}
-1, & k_1=1,3,... \\
 0, & k_1=2,4,... \\
\end{array}\right.
\end{equation}
and
\begin{equation}
F(K_2,c) = \left\{
\begin{array}{ll}
1, & k_2=1,3,... \\
2, & k_2=2,4,... \\
\end{array}\right.
\end{equation}
which yield
\begin{equation} \label{eq:case2-k}
F(K_2,c) - F(K_1,c) \geq 1.
\end{equation}

For $k_2=0$ which implies $G_2 \subset G_1$, since $G_2$ cannot be empty, $k_1
= g-|G_2|$ must be $\leq g$, there exists another permutation in $c$ such that
$\Omega_+(K_1,c)<\Omega_-(K_1,c)\leq \Omega_-(G_1,c)$ giving
\begin{equation}
F(K_1,c) = \left\{
\begin{array}{ll}
-1, & k_1=1,3,... \\
-2, & k_1=2,4,... \\
\end{array}\right.
\end{equation}
and $F(K_2,c) - F(K_1,c) \geq 1$ as in (\ref{eq:case2-k}) since $F(K_2,c)=0$.

Given that $G_2 = (G_1 \setminus K_1) \cup K_2$, with (\ref{eq:case2-g1}) and
(\ref{eq:case2-k}), we obtain
\[
F(G_2,c) = F(G_1,c) - F(K_1,c) + F(K_2,c) \geq 1
\]
and hence $\overrightarrow{G_2}(c)=1$. Since $\overrightarrow{G_1}(c)=0$, there
exists a column in $\mathbf{M_N}$ such that $\overrightarrow{G_2}(c) \neq
\overrightarrow{G_2}(c)$ which is sufficient to show that
$\overrightarrow{G_1}\neq\overrightarrow{G_2}$.

With the above two cases, we have proven that if $G_1 \neq G_2$, then
$\overrightarrow{G_1}\neq\overrightarrow{G_2}$. Together with the earlier
establishment that if $G_1 = G_2$, then $\overrightarrow{G_1} =
\overrightarrow{G_2}$, we conclude that for any $G_1,G_2 \subset G$ and
$G_1,G_2 \neq \emptyset$, $\overrightarrow{G_1} = \overrightarrow{G_2}$ iif
$G_1 = G_2$.

\vspace{-4.5mm}

\end{proof}

\vspace{4.5mm}

With Proposition~\ref{prop:code}, we have shown that based on our coding
scheme, a particular decoded bitstream uniquely identifies a particular
combination of receivers. This further allows the decoder to identify the
presence of each individual receiver in the superimposed transmission. In the
case that no receiver replies the acknowledgement, no transmission will occur
on the channel which indicates the absence of all receivers.

\section{Conclusion} \label{sec:conclusion}

Motivated by the concept of PNC and the idea of decoding ACKs for
improved broadcast reliability, we addressed the weaknesses of the
current designs and proposed a coding scheme that achieves robust
collision decoding without the need for precise energy detection. We
first established the majority principle between the input and the
output bitstreams based on BPSK modulation. Using the majority
principle, we developed a coding scheme achieving the identification
of individual acknowledgement transmission from a collided
transmission. We proved that our proposed coding scheme guarantees
the uniqueness of the decoded bitstream for a particular receiver
combination.

While our given approach has been demonstrated for BPSK modulation
scheme, it can similarly be extended to other modulation schemes
which satisfy the general PNC modulation-demodulation mapping
principle. Hence it is possible to implement a modified collision
coding schemes for M-ary PSK, PAM, M-ary QAM, (M-ary represent the
set of digital symbols) as well, which satisfy the general PNC
modulation-demodulation mapping principle. We expect this to be part
of our future research work.

% references section
% NOTE: BibTeX documentation can be easily obtained at:
% http://www.ctan.org/tex-archive/biblio/bibtex/contrib/doc/

% can use a bibliography generated by BibTeX as a .bbl file
% standard IEEE bibliography style from:
% http://www.ctan.org/tex-archive/macros/latex/contrib/supported/IEEEtran/bibtex
%\bibliographystyle{C:/jfcai/paper/Format/IEEE_Trans/IEEEtran}
% argument is your BibTeX string definitions and bibliography database(s)
%\bibliography{C:/jfcai/paper/Format/IEEE_Trans/IEEEabrv,C:/jfcai/paper/ref_network}
%
\bibliographystyle{IEEEtran}

%\bibliography{IEEEabrv,ref_network}
%
% <OR> manually copy in the resultant .bbl file
% set second argument of \begin to the number of references
% (used to reserve space for the reference number labels box)

\end{document}